\documentclass[10 pt, letterpaper]{article}

\usepackage{amsmath}
\usepackage{amssymb}
\usepackage{graphicx}
\usepackage{amsmath}
\usepackage{amssymb}
\usepackage{graphicx}
\usepackage{makeidx}
\usepackage{algorithm}
\usepackage{algpseudocode}
\usepackage{epstopdf}

\usepackage{times}
\usepackage[T1]{fontenc}

\newcommand{\eal}{\end{align}}
\newcommand{\beq}{\begin{equation}}
\newcommand{\eeq}{\end{equation}}
\newcommand{\bea}{\begin{eqnarray}}
\newcommand{\eea}{\end{eqnarray}}
\newcommand{\beas}{\begin{eqnarray*}}
\newcommand{\eeas}{\end{eqnarray*}}
\newcommand{\ba}{\begin{array}}
\newcommand{\ea}{\end{array}}

\newcommand{\bit}{\begin{itemize}}
\newcommand{\eit}{\end{itemize}}
\newcommand{\ben}{\begin{enumerate}}
\newcommand{\een}{\end{enumerate}}

\newcommand{\bR}{{\mathbb R}}
\newcommand{\bZ}{{\mathbb Z}}

\newtheorem{theorem}{Theorem}[section]
\newtheorem{proof}{Proof}
\newtheorem{assumption}{Assumption}

\newtheorem{definition}{Definition}[section]

\newtheorem{lemma}{Lemma}[section]
\newtheorem{problem}{Problem}[section]

\allowdisplaybreaks[3]

\makeatletter

\usepackage{makeidx}\usepackage{times}

\begin{document}
\date{}

\title{On-line direct data driven controller design approach with automatic update for some of the tuning parameters}
\author{M. Tanaskovic, L. Fagiano, C. Novara and  M. Morari \thanks{M. Tanaskovic and M. Morari are with the Automatic Control Laboratory, Swiss Federal Institute of Technology Zurich, Switzerland. L. Fagiano is with ABB Switzerland Ltd., Corporate Research,Baden-Daettwil, Switzerland and C. Novara is with the Dipartamento di Automatica e Informatica, Politechnico di Torino, Italy. Email addresses: tmarko@control.ee.ethz.ch (Marko Tanaskovic),  lorenzo.fagiano@ch.abb.com (Lorenzo Fagiano), carlo.novara@polito.it (Carlo Novara), morari@control.ee.ethz.ch (Manfred Morari). Corresponding author M.~Tanaskovic.}}

\maketitle

\section{Introduction}
This manuscript contains technical details of recent results developed by the authors on the algorithm for direct design of controllers for nonlinear systems from data that has the ability to to automatically modify some of the tuning parameters in order to increase control performance over time.

\section{Problem formulation}

\label{S:problem}

We consider a discrete, time invariant, nonlinear system with one input and $n_x$ states that can be represented by the following state equation:
\begin{equation}\label{Eq:sys}
x_{t+1}=g(x_t,u_t)+e_t,
\end{equation} 
where $t\in \bZ$ is the discrete time step, $u_t\in \bR$ is the control input, $x_t\in \bR^{n_x}$ is the vector of states and $e_t \in \bR^{n_x}$ is the vector of disturbance signals that accounts for the contribution of both the measurement and process disturbances. We make the following two assumptions about the disturbance signal $e_t$ and the nonlinear function $g$:
\begin{assumption}\label{A:bound}
The disturbance $e_t$ is bounded in magnitude:
\begin{equation}\label{Eq: noise_bound}
e_t \in B_{\epsilon}\doteq \{e_t:  || e_t||\leq \epsilon, \forall t\in \bZ\},
\end {equation}
for some $\epsilon >0$. 
\end{assumption}
\begin{assumption}\label{A:lip}
The function $g$ is Lipschitz continuous with respect to $u$, i.e. $g(x,\cdot)\in \mathcal{F}(\gamma_g, U)$ for any $x\in X$, where $U\subset \bR$ and $X\subset \bR^{n_x}$ are compact (possibly very large) sets, and
\begin{equation}
\begin{aligned}
\mathcal{F}(\gamma_g,\! U)\!\doteq \!\left\{\!\!\begin{aligned} g\!:&
||g(0)||\!<\!\infty,\\\!&||g(u_1)\!\!-\!\!g(u_2)||\!\leq\! \gamma_g
||u_1\!\!-\!\!u_2||,\forall u_1,\!u_2\!\in\! U \end{aligned}\!\right\}\!\!.
\end{aligned}
\end{equation}
\end{assumption}  
The notation $||\cdot ||$ stands for a suitable vector norm chosen by the user (typically $2$- or $\infty$-norm) and the presented results hold for any particular norm selection.

It is assumed that the nonlinear function $g$ that describes the dynamics of the system \eqref{Eq:sys} is unknown, but that a set $\mathcal{D}_N$ of $N$ noise corrupted input and state measurements generated by the system \eqref{Eq:sys} is available:
\beq\label{Eq:trdat}
\mathcal{D}_N\doteq\left\{u_t ,\omega_t \right\}_{t=-N}^{-1},\, \omega_t\doteq(x_{t}, x_{t+1}).
\eeq
We make the following assumption on the training data \eqref{Eq:trdat}.
\begin{assumption}\label{A:measurements}
The available measurements $\mathcal{D}_N$ are such that $u_t\in U$ and $w_t \in X\times X, \, \forall t=-N, \hdots, -1$.
\end{assumption}

In this note we consider the notion of finite gain stability.
\begin{definition}\label{D:FGS} (Finite gain stability)
A nonlinear system (possibly time varying) with input $u_t \in
\bR$, state $x_t \in \bR^{n_x}$ and disturbance $e_t \in B_{\epsilon}$ is \emph{finite gain stable} if there exist finite and nonnegative constants
$\lambda_1$, $\lambda_2$ and $\beta$ such that: \beq\label{Eq:stabs}
||\bold{x}||_{\infty}\leq \lambda_1 ||\bold{u}||_{\infty}+\lambda_2||\bold{e}||_{\infty}+\beta,
\forall \bold{u}\in \mathcal{U}, \forall \bold{e} \in
\mathcal{B}_{\epsilon}, \eeq where $\bold{x}=(x_1,x_2,\hdots)$,
$\bold{u}=(u_1,u_2,\hdots)$, $\bold{e}=(e_1,e_2,\hdots)$,
$||\bold{x}||_{\infty}\doteq \sup\limits_k||x_k||$ and $\mathcal{U}$
and $\mathcal{B}_{\epsilon}$  are the domains of the input and
disturbance signals, respectively.
\end{definition} 
Based on this definition, we introduce the notion of $\gamma$ stabilizability.
\begin{definition}\label{D:stabilizability}
The system \eqref{Eq:sys} is $\gamma$-\emph{stabilizable} if there exists a $\gamma<\infty$ and a function $f\in \mathcal{F}(\gamma,\bR^{2n_x})$ such that the closed-loop system with input $r_t\in B_r\subseteq X$ and disturbance $e_t\in \mathcal{B}_{\epsilon}$:
\beq\label{Eq:closed_loopsys}
x_{t+1}=g\left(x_t,f(x_t,r_{t+1})\right)+e_{t}
\eeq
is finite gain stable. 
\end{definition}

In Definition \ref{D:stabilizability}, the reference signal is assumed to belong to a compact $B_{\overline{r}}\subseteq X$, i.e. the reference is bounded in norm by the scalar $\overline{r}$ and it is never outside the set where the state trajectory shall be confined.

\begin{assumption}\label{A:stabilizability}
The system \eqref{Eq:sys} is $\gamma$-stabilizable for some $\gamma<\infty$.
\end{assumption}

We can finally state the problem that we address.
\begin{problem}\label{P:prob}
Use the batch of data $\mathcal{D}_N$, collected up to $t=0$, to design a feedback controller whose aim is to track a desired
reference signal $r_t\in B_{\overline{r}}$ for $t>0$. Once the controller is in operation, carry out on-line refinement of the design by exploiting the incoming input and state
measurements, while keeping the closed loop system finite gain stable.
\end{problem}

\section{On-line direct control design method}\label{S:Overall_algorithm}

We approach Problem \ref{P:prob} from the point of view of data-driven, direct dynamic inversion techniques. In this context, we assume the existence of an ``optimal'' (in a sense that will be shortly specified) inverse of the system's dynamics \eqref{Eq:sys} among the functions that, if used as controller, stabilize the closed-loop system. Then, we build from the available prior knowledge and data a set of functions that is guaranteed to contain the optimal inverse, and we exploit such a set to derive an approximated inverse, which we use as feedback controller. This approach involves several preliminary ingredients, explained in the following sub-sections.

\subsection{Optimal inverse and controller structure}
Following the definitions and notation introduced in \cite{lorenzo}, for a given control function $f$ we define the point-wise inversion error as:
\beq\label{Eq:PIE}
IE(f,r,x,e) \doteq \| r-g(x,f(r,x))-e\|,
\eeq
and the global inversion error as:
\beq\label{Eq:global}
GIE(f)=_L\|IE(f,\cdot,\cdot,\cdot)\|,
\eeq
where $_L\|\cdot\|$ in \eqref{Eq:global} is a suitable function norm (e.g. $L_\infty$) evaluated on $X\times B_{\overline{r}}\times B_{\epsilon}$. Based on Assumption \ref{A:stabilizability}, there exist a set $\mathcal{S}$ containing all functions $f$ that stabilize the closed loop system. Then, we define the optimal inverse controller function $f^*$ as:
\beq\label{Eq:opt_ctrl}
f^*=\text{arg} \min\limits_{\mathcal{S}\bigcap \mathcal{F}_{X\times X}}GIE(f),
\eeq
where $\mathcal{F}_{X\times X}$ denotes the set of all Lipschitz continuous functions on $X \times X$. We denote the Lipschitz constant of $f^*$ with $\gamma^*$, and the related constants $\lambda_1$, $\lambda_2$ and $\beta$, obtained if the controller $f^*$ were used in closed-loop (see \eqref{Eq:stabs}), by
$\lambda_1^*$, $\lambda_2^*$ and $\beta^*$.

Considering the measured data available up to a generic time $t$, we can write the control input as:
\beq\label{Eq:input_equation}
u_t=f^*(\omega_t)+d_t,
\eeq
where $d_t$ is a signal accounting for the unmeasured noise and disturbances and possible inversion errors. From Assumptions \ref{A:bound} and \ref{A:lip}, it holds that as long as the state and input trajectories evolve in the sets $X$ and $U$, respectively, the scalar $d_t$ has to be bounded, i.e. $d_t \in B_{\delta}\subset \bR$ with $\delta$ being a positive constant. Then, following a set membership identification approach (see e.g. \cite{traub22,nonlinear_milanese}), we consider the set of feasible inverse functions at time step $t$ ($FIFS_t$), i.e. the set of all functions $f\in \mathcal{F}_{X \times X}$ that are consistent with  the available data and prior information:
\beq\label{Eq:FIFS0}
 FIFS_t \doteq \bigcap\limits_{j=-N,\hdots,t-1}H_j,
\eeq
where:
\beq\label{Eq:FIFS}
H_j \doteq \{ f\in \mathcal{F}_{X\times X}:|u_j-f(\omega_j)|\leq \delta\}.
\eeq
The inequality in \eqref{Eq:FIFS} stems from the observation that the measured input $u_t$ and the value of function $f^*$ evaluated at the corresponding $\omega_t$ can not be larger than the bound on the amplitude of the signal $d_t$.

Under Assumptions \ref{A:bound}--\ref{A:stabilizability}, if $u_t\in U$ and $x_t \in X\times X,\, \forall t\geq -N$, then the optimal inverse $f^*$ belongs to $FIFS_t$, i.e. $f^*\in FIFS_t$ for all $t$. In set membership identification, an estimate $f\approx f^*$ belonging to the set $FIFS_t$ enjoys a guaranteed worst-case approximation error not larger than twice the minimal that can be achieved (see e.g. \cite{traub22} for details). Motivated by this accuracy guarantee, we update the approximation of the optimal inverse controller $f^*$, $f_t$ on-line in order to approach the set $FIFS_t$. First, in order to have a tractable computational problem, we  parameterize the controller $f_t$ with a finite sum of kernel functions:
\beq\label{Eq:parametrization}\nonumber
f_t(\omega)=a_t^TK(\omega,W_t),
\eeq
where $a_t \in \bR^{L_t}$
is the vector of weights, and $K(\omega,W_t)=\left[
\kappa(\omega,\tilde{\omega}_1),\hdots,
\kappa(\omega,\tilde{\omega}_{L_t}) \right]^T$ is a vector of kernel functions
$\kappa(\cdot,\tilde{\omega}_i):\bR^{2n_x}\rightarrow
\bR,i=1,\hdots,L_t$ belonging to a dictionary that is
uniquely determined by the $L_t$  kernel function centers
$W_t=\{\tilde{\omega}_1,\hdots, \tilde{\omega}_{L_t}\}$.
Then, at each time $t$ we update the set $W_t$, which determines the kernel function dictionary, and we also recursively update the weights $a_t$ exploiting the knowledge of $FIFS_t$ \eqref{Eq:FIFS0}, with an approach  inspired by the projection-based learning scheme presented in \cite{Proj_3} in the context of signal processing. Moreover, in order to achieve finite gain stability of the closed-loop system, we exploit the information that $f^*\in FIFS_0$ to derive a robust constraint on the vector of weights $a_t$, which we impose in the on-line procedure.

In the following, we provide the details of these steps and at the end we summarize the overall design method.

\subsection{Robust inequality to enforce closed loop stability}
We require the approximated inverse, $f_t$, to satisfy the following inequality at each time step $t\geq 0$:
\beq\label{Eq:stab_cond}
\begin{aligned}
&|f_t(\omega^+_t)-f^*(\omega^+_t)|\leq \gamma_{\Delta}\|x_t\|+\sigma, \\&\forall f^* \in \mathcal{F}(\gamma^*,X\times X) \cap FIFS_0,\, \forall t\geq 0,
\end{aligned}
\eeq
where $\omega^+_t=[x_t,r_{t+1}]^T$ and $\gamma_{\Delta}, \sigma \in \bR, \gamma_{\Delta}, \sigma>0$, are design parameters. Guidelines on how these parameters should be selected in order to guarantee finite gain stability of the closed-loop are given in Section \ref{S:Tuning}. The idea behind \eqref{Eq:stab_cond} is to limit the discrepancy between the input computed by the approximate inverse $f_t$ at time step $t$, i.e. $u_t=f_t(\omega^+_t)$, and the one given by the optimal inverse $f^*$ to a sufficiently small value, which depends linearly on the norm of the current state. However, since the optimal inverse $f^*$ is not known, we require the inequality \eqref{Eq:stab_cond} to be satisfied robustly for all functions in $FIFS_0$ that have the Lipschitz constant equal to $\gamma^*$. As mentioned above, such a function set is in fact guaranteed to contain $f^*$ under our working assumptions.

To translate the inequality \eqref{Eq:stab_cond} into a computationally tractable constraint on the parameters $a_t$, we exploit the information that $f^*\in\mathcal{F}(\gamma^*,X\times X) \cap FIFS_0$ to compute tight upper and lower bounds  on $f^*(\omega^+_t)$ using the following result from \cite{nonlinear_milanese}.
\begin{theorem}\label{T:milanese}(Theorem 2 in \cite{nonlinear_milanese})
Given a nonlinear function $f^*\in \mathcal{F}(\gamma^*,X\times X)\cap FIFS_0$, the following inequality holds:
\[
\underline{f}(\omega)\leq f^*(\omega)\leq \overline{f}(\omega),
\]
where:
\beq\label{Eq:bounds}
\begin{aligned}
\overline{f}(\omega)&=\min\limits_{k=-N,\hdots,-1}\left(u_k+\delta+\gamma^*\|\omega-\omega_k\|\right)\\
\underline{f}(\omega)&=\max\limits_{k=-N,\hdots,-1}\left(u_k-\delta-\gamma^*\|\omega-\omega_k\| \right).
\end{aligned}
\eeq
\end{theorem}
Exploiting Theorem \ref{T:milanese}, the robust constraint \eqref{Eq:stab_cond} can be satisfied by enforcing the following two inequalities on the vector of weights $a_t$:
\beq\label{Eq:robust_ineq}
\begin{aligned}
&-\gamma_{\Delta}\|x_t\|-\sigma+\overline{f}(\omega^+_t)\leq a_t^TK(\omega_t^+,W_t)\\
& a_t^TK(\omega_t^+,W_t) \leq \gamma_{\Delta}\|x_t\|+\sigma+\underline{f}(\omega^+_t)
\end{aligned}
\eeq
Note that the value of $\gamma^*$ that is required in order to enforce the constraints in \eqref{Eq:robust_ineq} needs to be estimated from the available measurement data (see also Section \ref{S:Tuning}). We describe next the approach to update the dictionary of kernel functions and the vector of weights $a_t$.

\subsection{Updating the dictionary of kernel functions}

The data generated by any Lipschitz continuous nonlinear function  evaluated at a finite number of points can be well approximated by a dictionary of kernel functions centered at the same points. In our on-line controller design, we let the dictionary grow and incorporate new kernel functions as new input and state measurements are collected. However, adding a new function to the dictionary at each time step would lead to an unlimited growth of the dictionary size $L_t$ over time. Moreover, this would result in a dictionary that is not sparse, i.e. with many functions that are similar (centered at points close to each other), and with possible over-fitting of the measurement data. To avoid these problems, we choose to add a new function only if it is sufficiently different from those already contained in the dictionary. As indicator of similarity, we use the so-called coherence factor (see e.g. \cite{koherence} for more details):
\beq\label{Eq:koherence}
\mu(\omega,W_t)=\max\limits_{i=1,\hdots, L_t}|\kappa(\omega,\tilde{\omega_i})|.
\eeq
Note that $\mu(\omega,W_t)\in (0,1]$, and that $\mu(\omega,W_t)=1$ if and only if $\omega \in W_t$. Hence, the larger the coherence value in \eqref{Eq:koherence}, the more similar is the kernel function centered at $\omega$ to some function already in the dictionary. In our design technique, we set a threshold $\overline{\mu}\in (0,1)$ and we add a particular data point $\omega$ to the set of function centers $W_t$ only if $\mu(\omega,W_t)\leq \overline{\mu}$. This approach guarantees that the size of the dictionary will remain bounded over time (see e.g. \cite{koherence}).

\subsection{Updating the vector of weights}

As a preliminary step to the recursive update of the weights $a_t$ note that, as discussed above, the size of the dictionary can expand from time step $t-1$ to time step $t$ and therefore in general it will hold that $a_{t-1}\in \bR^{L_{t-1}}$ and $a_t \in \bR^{L_t}$ with $L_{t-1}\leq L_t$. Therefore, in order to properly define the updating algorithm at time $t$, we consider the vector $a^+_{t-1}\in \bR^{L_t}$:
\beq\label{Eq:vector_ext}
a^+_{t-1}=[a_{t-1}^T, \underbrace{0,\hdots,0}_{L_t-L_{t-1}}]^T,
\eeq
obtained by initializing the weights corresponding to the kernel functions that are added to the dictionary to zero.

To introduce the updating of the vector $a_t \in \bR^{L_t}$, we note that each pair $(u_j,\omega_j),\,j=-N,\ldots,t-1$ defines, together with the dictionary of kernel functions at time step $t$, the following set:
\beq\label{Eq:measurement_strip}
S_{jt}\doteq\{ a \in\bR^{L_t}: |a^T\kappa(\omega_j,W_t)-u_j|\leq \delta\},
\eeq
which is a strip (hyperslab) in $\bR^{L_t}$. If $a_t\in S_{jt}$, then  the corresponding function $f_t$ in \eqref{Eq:parametrization} belongs to the set $H_j$ defined in \eqref{Eq:FIFS}. We further define the projection of a point in $\bR^{L_t}$ onto the strip $S_{jt}$ as:
\beq\label{Eq:measurement_projection}
P_{jt}(a)\doteq\min\limits_{\hat{a}\in S_{jt}} \|a-\hat{a}\|_2.
\eeq
Note that calculating the projection \eqref{Eq:measurement_projection} amounts to solving a very simple linear program, whose solution can be explicitly derived (see e.g. \cite{theodoridis_rev}). Therefore, calculating the projection of any point in $\bR^{L_t}$ onto a measurement strip  as in \eqref{Eq:measurement_projection} can be done computationally very efficiently. Finally, we consider the hyperslab defined by the stability constraint \eqref{Eq:robust_ineq}:
 \beq\label{Eq:stability_strip}
S^+_t\!\doteq\!\left\{\!\begin{aligned}a \in
\bR^{L_t}\!\!:\;&a^TK(\omega^+_t,W_t) \geq-\!\gamma_{\Delta}\|x_t\|\!-\!\sigma\!+\!\overline{f}(\omega^+_t)
\\&a^TK(\omega^+_t,W_t) \! \leq \!
\gamma_{\Delta}\|x_t\|\!+\!\sigma\!+\!\underline{f}(\omega^+_t) \end{aligned}\!\right \}\!,
\eeq
and we denote the corresponding projection operator with $P^+_t(\cdot)$.

From the definitions of the hyperslabs $S_{jt}$ and $S_t^+$ in \eqref{Eq:measurement_strip} and \eqref{Eq:stability_strip} it follows that if \[
a_t\in  S_t^+\bigcap\left( \bigcap_{j=-N,\hdots,t-1}S_{jt}\right),
\]
then the corresponding function $f_t$ belongs to the set $FIFS_t$ and satisfies the stabilizing constraint \eqref{Eq:stab_cond}. However, to find a point that belongs to the intersection of all $S_{jt},\, j=-N,\hdots,t-1$ at each time step is computationally challenging. Therefore, we exploit the idea at the basis of projection learning algorithms, that by repeatedly applying the projection operators to a point, the result will eventually fall in the intersection of the considered hyperslabs. In particular, we update the vector of weights $a_t$ in two steps:  first, following the idea of \cite{Proj_3}, we calculate a convex combination of its projections onto the hyperslabs defined by a finite number $q\geq1$ of the latest measurements; then, we project the obtained point onto the hyperslab $S_t^+$ in order to ensure the satisfaction of the stabilizing constraint \eqref{Eq:stab_cond}.
To be more specific, let the set of indexes $J_t=\{\max\{-N,t-q\}, \hdots, t-1 \}$ contain the time instants of the last $q$ state and input measurements, and let $I_t=\{j\in J_t: a_{t-1}^+\notin S_{jt} \}$ be the subset of indexes such that the weighting vector $a^+_{t-1}$ does not belong to the corresponding hyperslabs. Then, we compute our update of the weighting vector $a_t$ from $a_{t-1}^+$ as:
\beq\label{Eq:projection_algorithm}
a_{t}=P^+_t\left(a^+_{t-1}+ \sum_{j\in I_t} \frac{1}{card(I_t)} \left(P_{jt}(a^+_{t-1})-a^+_{t-1}\right)\right),
\eeq
where $card(I_t)$ denotes the number of elements in $I_t$. This update  can be computed very efficiently with the explicit formulas for vector projections and eventually by parallelizing the projection operations.

\subsection{Summary of the proposed design algorithm}
The described procedures to update the dictionary of kernel functions and the weights $a_t$ form our on-line scheme to compute the feedback controller $f_t$, summarized in Algorithm \ref{A:learning_algorithm}. 

\begin{algorithm}
\begin{itemize}
\item[1)] Collect the state measurement $x_t$. If $t<0$, set $\omega^+_t=[x_t,x_{t+1}]^T$, otherwise set $\omega^+_t=[x_t,r_{t+1}]^T$.
\item[2)] Update the dictionary $W_t$ starting from $W_{t-1}$ and adding $\omega^+_t$ if $\mu(\omega^+_t,W_{t-1})\leq \overline{\mu}$ and $\omega_{t-1}$ if $\mu(\omega_{t-1}, W_{t-1})\leq \overline{\mu}$. Form the vector $a^+_{t-1}$ according to \eqref{Eq:vector_ext}.
\item[3)] Calculate $a_t$ according to \eqref{Eq:projection_algorithm}.
\item[4)] If $t\geq 0$, calculate the input $u_t=a_t^TK(\omega^+_t,W_t)$ and apply it to the plant.
\item[5)] Set $t=t+1$ and go to $1)$.
\end{itemize}
\caption{Feedback control algorithm based on the on-line direct control design scheme}\label{A:learning_algorithm}
\end{algorithm}
For $t\geq0$, such an algorithm is both a controller and a design algorithm, while for $t<0$ it only acts as a design algorithm.
\section{Algorithm tuning}\label{S:Tuning}
In order to implement Algorithm \ref{A:learning_algorithm}, several tuning parameters need to be selected. These are the noise bound $\delta$ and the Lipshitz constant $\gamma^*$ which are required for calculating the projections on the hyperslabs $S_{jt}$ and $S_t^+$. In addition parameters $\gamma_{\Delta}$ and $\sigma$ in \eqref{Eq:robust_ineq} need to be selected. Careful selection of these parameters guarantees finite gain stability of the closed loop. Namely the parameters $\gamma_{\Delta}$ and $\sigma$ should be selected such that:
\beq\label{Eq:cond_gama}
\gamma_{\Delta}\in \left(0,\frac{1}{\gamma_g\lambda_2^*}\right),
\eeq
and
\beq\label{Eq:cond_sigma}
\sigma>\frac{D_0}{2},
\eeq
where 
\beq
D_0
\doteq \sup_{\omega \in B_{\overline{xr}}}\left (\overline{f}(\omega)
-\underline{f}(\omega)\right),
\eeq
with
\beq
B_{\overline{x}\overline{r}} \doteq \left\{\omega \in \bR^{n_x}
\times R: \omega=(x,r), \forall x\in B_{\overline{x}}, r\in \mathcal{B}_{\overline{r}}
\right\}.
\eeq 
and with $\overline{x}$ given as
\beq\label{Eq:max_x}
\overline{x}\doteq\frac{\lambda_1^*\overline{r}+\gamma_g\lambda_2^*\sigma+\lambda_2^*\epsilon+\beta^*}{1-\gamma_g\lambda_2^*\gamma_{\Delta}},
 \eeq 

In order to verify whether \eqref{Eq:cond_gama} and \eqref{Eq:cond_sigma} hold, the values of $\epsilon$, $\delta$, $\gamma_g$, $\gamma^*$, $\lambda_1^*$, $\lambda_2^*$ and $\beta^*$ should be known. The values of $\lambda_1^*$, $\lambda_2^*$ and $\beta^*$  can not be
estimated based on the available data and they have to be guessed. Since these parameters are related to the performance of the optimal inverse controller $f^*$ that should typically result in small tracking error (i.e. the state of the corresponding closed loop system should be close to the desired reference signal), a reasonable guess for $\lambda_1^*$ and $\lambda_2^*$ is a value slightly greater than $1$ and for $\beta^*$ a value close to $0$. The Lipschitz constants $\gamma_g$ and $\gamma^*$ as well as the disturbance bounds $\epsilon$ and $\delta$ can be estimated from the
available training data (see e.g. \cite{lorenzo}). In order to be able to use the parameter estimates that might be subject to estimation errors while still retaining the stability guarantees, the obtained estimates can be inflated by positive constants that reflect the level of estimation uncertainty, i.e. the parameters $\epsilon$, $\delta$, $\gamma_g$ and $\gamma^*$ can be selected as: $\epsilon=\hat{\epsilon}_{-1}+c_{\epsilon}$, $\delta=\hat{\delta}_{-1}+c_{\delta}$, $\gamma_g=\hat{\gamma}_{g,-1}+c_{\gamma_g}$ and $\gamma^*=\hat{\gamma}^*_{-1}+c_{\gamma^*}$, where $\hat{\epsilon}_{-1}$, $\hat{\delta}_{-1}$, $\hat{\gamma}_{g,-1}$ and $\hat{\gamma}^*_{-1}$ denote the parameter estimates obtained by applying the algorithms described in \cite{lorenzo} to the training data $\mathcal{D}_N$ and $c_{\epsilon}$, $c_{\delta}$, $c_{\gamma_g}$ and $c_{\gamma^*}$ are positive constants that should be selected by the control designer and that should reflect his feeling on the size of the possible estimation error. Hence, in order to ensure the satisfaction of \eqref{Eq:cond_gama} and \eqref{Eq:cond_sigma}, parameters $\gamma_{\Delta}$ and $\sigma$ can be selected such that $\gamma_{\Delta} \in \left( 0,\frac{1}{(\hat{\gamma}_{g,-1}+c_{\gamma_g})\lambda_2^*}\right)$ and
\beq\label{second_rule_sig}
\sigma \geq \frac{1}{2} \sup\limits_{\omega\in B_{\overline{xr}}}(\overline{f}_c(\omega)-\underline{f}_c(\omega)),
\eeq 
where
\beq\label{Eq:bounds_tv1}
\begin{aligned}
\overline{f}_c(\omega)\!&=\!\min\limits_{k=-N,\hdots,-1}\left(\!u_k\!+\!\hat{\delta}_{-1}\!+\!c_{\delta}\!+\!(\hat{\gamma}^*_{-1}\!+\!c_{\gamma^*})||\omega\!-\!\omega_k||\right)\\
\underline{f}_c(\omega)\!&=\!\max\limits_{k=-N,\hdots,-1}\left(\!u_k\!-\!\hat{\delta}_{-1}\!-\!c_{\delta}\!-\!(\hat{\gamma}^*_{-1}\!+\!c_{\gamma^*})||\omega\!-\!\omega_k|| \right)\!.
\end{aligned}
\eeq

However, some of these parameters can also be updated on-line, which should increase their accuracy as new input and state measurements are collected and hence increase the overall performance of the controller. In the following section we describe how some of the tuning parameters can be updated on-line and we prove the finite gain stability of the closed loop system in this case.

\subsection{On-line adaptation of some of the tuning parameters}\label{S:basis_fcn}

Note that the value of $\epsilon$ is required for selecting the tuning parameter $\sigma$. Recalculating the value of $\sigma$ that satisfies \eqref{Eq:cond_sigma} over time would be computationally demanding and therefore we do not consider updating of the parameter $\sigma$ and the noise bound estimate $\epsilon$ over time. Hence the parameter $\epsilon$ can be selected by properly inflating the estimate $\hat{\epsilon}_{-1}$ obtained from the initially available training data and $\sigma$ can be selected such that it satisfies \eqref{second_rule_sig}. On the other hand selecting $\gamma_{\Delta}$ that satisfies \eqref{Eq:cond_gama} based on $\gamma_g$ is very easy and therefore we will consider its modification over time. Moreover, 
the updating of parameters $\delta$ and $\gamma^*$ that influence the projections done under Algorithm \ref{A:learning_algorithm} is also considered.

 To this end, we consider a generic nonlinear function $f': \bR^{n_{\xi}}\to\bR^{n_z}$ with $n_{\xi}$ inputs and $n_z$ outputs that is Lipschitz continuous with the constant $\gamma$ and whose output is corrupted by disturbances $o_t$ as:
\beq\label{Eq:generic}
z_t=f'(\xi_t)+o_t,
\eeq
where $o_t \in \bR^{n_z},\,|o_t|\leq \varepsilon,\forall t$. We introduce two on-line algorithms for updating the estimates of the noise bound and the Lipschitz constant over time, that we denote by $\hat{\varepsilon}_t$ and $\hat{\gamma}_t$ at time step $t$ respectively.\\

\begin{algorithm}
\begin{itemize}
\item[1)] Chose a ``small`` $\rho>0$. For example $\rho=0.01 \max\limits_{i,j=-N,\hdots,-1} ||\xi_i-\xi_j||$, initialize $\hat{\varepsilon}_{-N}$ to $0$ and set $t=-N+1$.
\item[2)] Find the set of indexes: $J_t=\left\{k\in [max\{-N,t-\overline{N}\},\hdots, t]: ||\xi_t-\xi_k||\leq \rho  \right\}$
\item[3)] If $J_t=\emptyset$ set $\varepsilon_{z}=0$, otherwise set $\varepsilon_{z}=\frac{1}{2}\max\limits_{i\in J_t}||z_t-z_i||$.
\item[4)] Calculate $\hat{\varepsilon}_t=\max\{\hat{\varepsilon}_{t-1},\varepsilon_z\}$.
\item[5)] Set $t=t+1$ and go to $2)$. 
\end{itemize}
\caption{on-line estimation of the noise bound}\label{A:noise_bound}
\end{algorithm}

\begin{algorithm}
\begin{itemize}
\item[1)] Initialize $\hat{\gamma}_{-N}$, $\gamma_{-N}^{\text{current}}$ and $\Delta_{-N}^{\text{current}}$ to $0$ and set $t=-N+1$.
\item[2)] Calculate $\Delta_{kt}=||\xi_k-\xi_t||, k=\max\{-N,t-\overline{N}\},\hdots,t-1$.
\item[3)] For $k=\max\{-N,t-\overline{N}\},\hdots,t-1$ and $\Delta_{kt}\neq0$  calculate: 
\beq
\tilde{\gamma}^t_{kt}=\begin{cases} \frac{||z_t-z_k||-2\hat{\varepsilon}_t}{\Delta_{kt}} & \text{if} \,\, ||z_t-z_k||>2\hat{\varepsilon}_t\\ 0 & \text{if}\,\, \text{otherwise}\end{cases} 
\eeq
If $\Delta_{kt}=0$ set $\tilde{\gamma}^t_{kt}=0$.
\item[4)] For $j=\max\{-N,t-\overline{N}\},\hdots,t-1$ and $i=\max\{-N,t-\overline{N}\},\hdots,j-1$ calculate: $\tilde{\gamma}^t_{ij}=\tilde{\gamma}^{t-1}_{ij}-\frac{2(\hat{\varepsilon}_t-\hat{\varepsilon}_{t-1})}{\Delta_{ij}}$. In addition, calculate $\gamma^{\text{current}}_t=\gamma^{\text{current}}_{t-1}-\frac{2(\hat{\varepsilon}_t-\hat{\varepsilon}_{t-1})}{\Delta_{t-1}^{\text{current}}}$. If $\Delta_{ij}=0$ or $\Delta_{t-1}^{\text{current}}=0$, set $\tilde{\gamma}^t_{ij}=\tilde{\gamma}^{t-1}_{ij}$ and $\gamma^{\text{current}}_t=\gamma^{\text{current}}_{t-1}$.
\item[5)] Calculate:
\beq
\tilde{\gamma}_t=\max\limits_{\begin{array}{cc}j=\max\{-N,t-\overline{N}\},\hdots,t\\ i=\max\{-N,t-\overline{N}\},\hdots,j-1 \end{array}} \{ \tilde{\gamma}^t_{ij}\}
\eeq
\beq
\hat{\gamma}_t=\max\{\tilde{\gamma}_t, \gamma_t^{\text{current} } \}.
\eeq
If $\hat{\gamma}_t=\tilde{\gamma}_t$, set $\Delta_t^{\text{current}}= \Delta_{pq}$, where $\tilde{\gamma}_t=\tilde{\gamma}^t_{pq}$. Otherwise set $\Delta_t^{\text{current}}=\Delta_{t-1}^{\text{current}}$. Set $\gamma_t^{\text{current}}=\hat{\gamma}_t$.
\item[6)] Set $t=t+1$ and go to $2)$. 
\end{itemize}
\caption{on-line estimation of the Lipschitz constant}\label{A:Lp_const}
\end{algorithm}

Algorithms \ref{A:noise_bound} and \ref{A:Lp_const} can be run consecutively in order to update the estimates of $\epsilon$ and $\gamma$ at each time step. These two algorithms can be seen as on-line versions of the estimation methods proposed in \cite{lorenzo} and \cite{lorenzo_archive}. In order to limit the memory requirements of the proposed algorithms, we introduce a memory horizon $\overline{N}$ which denotes the maximal number of measurement points that need to be stored in memory during the execution of the algorithms. By setting $\overline{N}=N$, the functionality of the proposed algorithms becomes equivalent to the functionality of their off-line counterparts.

Algorithm \ref{A:noise_bound} can be used in order to estimate the noise bound $\delta$. We denote this estimate at time step $t$ by $\hat{\delta}_t$. The estimate $\hat{\delta}_t$ and the Algorithm \ref{A:Lp_const} can be used in order to update the estimate of the Lipschitz constant $\gamma^*$, that we denote by $\hat{\gamma}^*_t$. This can be done by setting the function $f'$ in \eqref{Eq:generic} to $f$, $z_t$ to $u_t$ and $\xi_t$ to $\omega_t$. The Lipschitz constant of function $g$ with respect to the input $u$,  $\gamma_g$, can be estimated by considering a reformulation of the state evolution equation \eqref{Eq:sys} given by:
\beq\label{Eq:alternative}
x_{t+1}=g'(u_t)+\vartheta_t,
\eeq
where $g'\doteq g(x^*,u_t)$ is the unknown function with the Lipschitz constant $\gamma_g$ and $x^*$ are given by:
\beq
\begin{aligned}
x^*&=\text{arg}\max\limits_{x\in B_{\overline{x}}} \mathcal{L}_g(x)\\
\mathcal{L}_g(x)&=\max\limits_{u_1,u_2\in U} \frac{||g(x,u_1)-g(x,u_2)||}{|u_1-u_2|},
\end{aligned}
\eeq
and 
\beq\label{Eq:new_noise}
\vartheta_t=g(x_t,u_t)-g(x^*,u_t)+e_t
\eeq
is an unknown disturbance signal. From the Assumptions \ref{A:bound} and \ref{A:lip} it follows that $\vartheta_t$ is bounded if $u_t \in U$ and $x_t\in X, \forall t\geq -N$, i.e. $\vartheta_t\in B_{\zeta}, \forall t\geq -N$, where $B_{\zeta}=\{\vartheta\in \bR^{n_x}:||\vartheta||_{\infty}\leq \zeta\}$. Therefore, Algorithm \ref{A:noise_bound} can be used to recursively update the bound of the disturbance $\vartheta_t$ and then Algorithm \ref{A:Lp_const} can be used in order to estimate the Lipschitz constant $\gamma_g$. In this case the function $f'$ in \eqref{Eq:generic} should be set to $g'$, $z_t$ should be set to $x_t$ and $\xi_t$ to $u_t$. We will denote the estimate of the Lipschitz constant $\gamma_g$ at time step $t$ by $\hat{\gamma}_{g,t}$.

In order to state the conditions under which the Algorithms \ref{A:noise_bound} and \ref{A:Lp_const} can be used together with the Algorithm \ref{A:learning_algorithm} in order to update some of its tuning parameters on-line, while preserving the finite gain stability of the closed loop system, we introduce the following assumption on the initially available data set $\mathcal{D}_N$. 

\begin{assumption}\label{A:excit}
Initially collected training data $\mathcal{D}_N$ are such that as $N\to \infty$, for any pair $(x,d)\in B_{\overline{x}}\times B_{\delta}$ and any $\lambda \in \bR,\lambda\geq0$ there exist a finite $N_{\lambda}\in \mathbb{N}, N_{\lambda}<\infty$ such that a pair $(y_t,d_t), t\in [-N_{\lambda},-1]$ satisfying $||(x,d)-(x_t,d_t)||\leq \lambda$ exists. In addition, for any pair $(u,\vartheta)\in U\times B_{\zeta}$ and any $\theta \in \bR,\theta\geq 0$ there exists a finite $N_{\theta}\in \mathbb{N}, N_{\theta}<\infty$ such that the pair $(u_t,\vartheta_t), t\in [-N_{\theta},-1]$ satisfying $||(u,\vartheta)-(u_t,\vartheta_t)||\leq \theta$ exists. 
\end{assumption}
This assumption ensures that the initial training data set is generated in such a way that the plant state $x_t$ and the disturbance signal $d_t$, as well as the input signal $u_t$ and the disturbance $\vartheta_t$ explore their domains well, i.e. the initially collected data is {\em informative enough}. Based on this assumption, we state the following Lemma on the properties of estimates $\hat{\delta}_t$, $\hat{\gamma}_t^*$ and $\hat{\gamma}_{g,t}$ obtained by using the Algorithms \ref{A:noise_bound} and \ref{A:Lp_const} on the training data $\mathcal{D}_N$, which is a direct consequence of the fact that the functionality of the introduced algorithms become equivalent to the functionality of off-line methods developed in \cite{lorenzo_archive} when $\overline{N}$ is set to $N$.
\begin{lemma}\label{L:convergence}
Let the Assumption \ref{A:excit} hold and let $\overline{N}=N$, with $N\to \infty$. If the algorithms \ref{A:noise_bound} and \ref{A:Lp_const} are used to estimate $\delta$, $\gamma^*$ and $\gamma_g$, it holds that $\hat{\delta}_t\to \delta$, $\hat{\gamma}^*_t \to \gamma^*$ and $\hat{\gamma}_{g,t}\to \gamma_g$ as $t\to \infty$. 
\end{lemma}
Therefore, if the training data would be infinitely long and sufficiently informative, then the estimates $\hat{\delta}_t$, $\hat{\gamma}^*_t$ and $\hat{\gamma}_{g,t}$ would converge to the corresponding true values. Based on this, we state the following Lemma that gives bounds on the estimation error of $\hat{\delta}_t$, $\hat{\gamma}^*_t$ and $\hat{\gamma}_{g,t}$ for $t\geq 0$ in the case when a finite number of training data is available.
\begin{lemma}\label{L:est_bounds}
Let the Assumption \ref{A:excit} hold. For any $c_{\delta}$, $c_{\gamma^*}$ and $c_{\gamma_g}\in \bR$ such that  $c_{\delta}\in(0,\delta)$, $c_{\gamma^*}\in(0, \gamma^*)$ and $c_{\gamma_g}\in(0, \gamma_g)$, there exists a finite $\tilde{N}\in \mathbb{N},\tilde{N}<\infty$ such that by setting $\overline{N}\geq N\geq \tilde{N}$ it holds that $\delta -c_{\delta}\leq \hat{\delta}_t \leq \delta$, $\hat{\gamma}_{g,t}\geq \gamma_g-c_{\gamma_g}$ and $\hat{\gamma}^*_t\geq \gamma^*-c_{\gamma^*}, \forall t\geq 0$ when Algorithms  \ref{A:noise_bound} and \ref{A:Lp_const} are used.
\end{lemma} 
\begin{proof}
We first note that, due to the step 3) of Algorithm \ref{A:noise_bound}, it holds that the estimates of the noise bounds $\delta$ and $\zeta$ can only increase over time, i.e. $\hat{\delta}_{t+1}\geq \hat{\delta}_t$ and $\hat{\zeta}_{t+1}\geq \hat{\zeta}_t, \forall t$. From the fact that $\hat{\delta}_t$ and $\hat{\zeta_t}$ are calculated by taking the maximum over the noise bound evaluated for individual data points, it holds that $\hat{\delta}_t\leq \delta$ and $\hat{\zeta}_{t}\leq \zeta, \forall t$. In addition, from Lemma \ref{L:convergence} it follows that for any $c_{\delta}, c_{\zeta}\in \bR$ such that $c_{\delta}\in (0,\delta), c_{\zeta}\in (0,\zeta)$, there exist finite $N_{\delta},N_{\zeta}\in \mathbb{N},N_{\delta},N_{\zeta}<\infty$ such that by setting $\overline{N}\geq N\geq \max\{N_{\delta},N_{\zeta}\}$, it holds that $\hat{\delta}_{-1}\geq \delta-c_{\delta}$ and $\hat{\zeta}_{-1}\geq \zeta - c_{\zeta}$. Therefore if $\overline{N}\geq N\geq \max\{N_{\delta},N_{\zeta}\}$, it has to hold that $\delta -c_{\delta}\leq \hat{\delta}_t \leq \delta$ and $\zeta -c_{\zeta}\leq \hat{\zeta}_t \leq \zeta$. Moreover, due to the steps 4) and 5) of Algorithm \ref{A:Lp_const}, it holds that if $\overline{N}\geq N\geq \max\{N_{\delta},N_{\zeta}\}$, then $\hat{\gamma}^*_t \geq \hat{\gamma}^*_{-1}-\frac{2c_{\delta}}{\Delta^{\text{current}}_{-1,\gamma^*}}$ and $\hat{\gamma}_{g,t} \geq \hat{\gamma}_{g,-1}-\frac{2c_{\zeta}}{\Delta^{\text{current}}_{-1,\gamma_g}}$, where $\Delta^{\text{current}}_{-1,\gamma^*}$ and $\Delta^{\text{current}}_{-1,\gamma_g}$ denote the value of $\Delta^{\text{current}}_{t}$ obtained at time step $t=-1$ when Algorithm \ref{A:Lp_const} is used for estimating $\gamma^*$ and $\gamma_g$ respectively. In addition, from Lemma \ref{L:convergence} it holds that for any $c_{\gamma^*}'=c_{\gamma^*}- \frac{2c_{\delta}}{\Delta^{\text{current}}_{-1,\gamma^*}}$ and $c_{\gamma_g}'=c_{\gamma_g}-\frac{2c_{\zeta}}{\Delta^{\text{current}}_{-1,\gamma_g}}$, there exist $N_{\gamma^*}',N_{\gamma_g}'\in \mathbb{N}, N_{\gamma^*}',N_{\gamma_g}'<\infty$ such that by setting $\overline{N}\geq N \geq \max\{N_{\gamma^*}',N_{\gamma_g}'\}$ it holds that $|\hat{\gamma}^*_{-1}-\gamma^*|\leq c_{\gamma^*}'$ and $|\hat{\gamma}_{g,-1}-\gamma_g|\leq c_{\gamma_g}'$. Therefore by setting $\overline{N}\geq N \geq \tilde{N}=\max\{ N_{\delta}, N_{\zeta}, N_{\gamma^*}',N_{\gamma_g}'\}$ it will hold that $\delta -c_{\delta}\leq \hat{\delta}_t \leq \delta$, $\hat{\gamma}_{g,t}\geq \gamma_g-c_{\gamma_g}$ and $\hat{\gamma}^*_t\geq \gamma^*-c_{\gamma^*}, \forall t\geq 0$. \hfill$\blacksquare$
\end{proof}
Hence according to Lemma \ref{L:est_bounds}, if the training data set $\mathcal{D}_N$ is informative and long enough, bounds on the accuracy of the estimates $\hat{\delta}_t$, $\hat{\gamma}_t^*$ and $\hat{\gamma}_{g,t}$ are guaranteed $\forall t$. 

In analogy to the definition of the hyperslabs $S_{jt}$ and $S_t^+$  in \eqref{Eq:measurement_strip} and \eqref{Eq:stability_strip}, we define the hyperslabs $\hat{S}_{jt}$ and $\hat{S}_t^+$ that depend on the time varying estimates $\hat{\delta}_t$ and $\hat{\gamma}_t^*$ and the time varying value of the tuning parameter $\gamma_{\Delta}$ that we denote by $\gamma_{\Delta,t}$ as:
\beq
\hat{S}_{jt}\doteq\{ a \in\bR^{L_t}: |a^T\kappa(\omega_j,W_t)-u_j|\leq \hat{\delta}_t\},
\eeq
\small
\beq\label{Eq:stability_stripv}
\hat{S}^+_t\!\doteq\!\left\{\!\begin{aligned}a \in
\bR^{L_t}:&-\!\gamma_{\Delta,t}||x_t||\!-\!\sigma\!+\!\overline{f}_t(\omega^+_t)\!\leq\!
a^TK(\omega^+_t,W_t) \\&a^TK(\omega^+_t,W_t) \! \leq \!
\gamma_{\Delta,t}||x_t||\!+\!\sigma\!+\!\underline{f}_t(\omega^+_t) \end{aligned}\!\right \}\!,
\eeq
\normalsize
where 
\beq\label{Eq:bounds2}
\begin{aligned}
\overline{f}_t(\omega)&\doteq \min\limits_{k=-N,\hdots,-1}\left(u_k\!+\!\hat{\delta}_{-1}\!+\!c_{\delta}\!+\!\hat{\gamma}||\omega\!-\!\omega_k||\right)\\
\underline{f}_t(\omega)&\doteq \max\limits_{k=-N,\hdots,-1}\left(u_k\!-\!\hat{\delta}_{-1}\!-\!c_{\delta}\!-\!\hat{\gamma}||\omega\!-\!\omega_k||  \right),
\end{aligned}
\eeq
and $\hat{\gamma}\doteq\min\{\hat{\gamma}^*_t\!+\!c_{\gamma^*},\hat{\gamma}^*_{-1}\!+\!c_{\gamma^*}\}$, with $\hat{\delta}_{-1}$ and $\hat{\gamma^*}_{-1}$ being the estimates of $\delta$ and $\gamma^*$ obtained from the training data $\mathcal{D}_N$ either by using Algorithms \ref{A:noise_bound} and \ref{A:Lp_const} or the method proposed in \cite{lorenzo} and $c_{\delta}\in (0,\delta)$ and $c_{\gamma^*}\in (0,\gamma^*)$ are design parameters. Based on this, we define the projection update equation that should be used by Algorithm \ref{A:learning_algorithm} instead of \eqref{Eq:projection_algorithm} when the Algorithms \ref{A:noise_bound} and \ref{A:Lp_const} are used to update the estimates of $\delta$, $\gamma^*$ and $\gamma_g$ over time:
\beq\label{Eq:projection_algorithm2}
a_{t}=\hat{P}^+_t\left(a^+_{t-1}+\sum_{j\in I_t} \frac{1}{card(I_t)} \left(\hat{P}_{jt}(a^+_{t-1})-a^+_{t-1}\right)\right),
\eeq
where $\hat{P}^+_t(\cdot )$ and $\hat{P}_{jt}(\cdot )$ denote projection operators onto hyperslabs $\hat{S}_t^+$ and $\hat{S}_{jt}$.

In addition, in order to account for the fact that the parameter $\gamma_{\Delta}$ can change over time, we redefine the maximal achievable state amplitude as
\beq\label{eq:other_bound}
\overline{x}\doteq \frac{\lambda_1^*\overline{r}+(\hat{\gamma}_{g,-1}+c_{\gamma_g})\lambda_2^*\sigma+\lambda_2^*\epsilon+\beta^*}{1-(\hat{\gamma}_{g,-1}+c_{\gamma_g})\lambda_2^*\overline{\gamma}_{\Delta}},
\eeq
where $\hat{\gamma}_{g,-1}$ is the estimate of $\gamma_g$ obtained from the initially available training data and $\overline{\gamma}_{\Delta}$ is a design parameter that should be selected such that $\overline{\gamma}_{\Delta}<\frac{1}{(\hat{\gamma}_{g,-1}+c_{\gamma_g})\lambda_2^*}$. 

Moreover, in analogy to \eqref{Eq:cond_gama}, we make the following assumptions about the selection of the tuning parameter $\hat{\gamma}_{\Delta,t}$.
\beq\label{Eq:time_varpar}
\hat{\gamma}_{\Delta,t}\in \left(0, \min\left\{\frac{1}{(\hat{\gamma}_{g,t}+c_{\gamma_g})\lambda_2^*},\overline{\gamma}_{\Delta}\right\}\right).
\eeq
In order to state the result on finite gain stability of the closed loop we make an additional assumption about the balls $B_{\overline{x}}$ and $B_{\overline{xr}}$ defined by $\overline{x}$ in \eqref{eq:other_bound}.
\begin{assumption}\label{A:contained}
$B_{\overline{x}}\subseteq X$.  Moreover, $\forall \omega \in
B_{\overline{xr}}, \forall \Delta u \in
[-\gamma_{\Delta}\overline{x}-\sigma,\gamma_{\Delta}\overline{x}+\sigma],
f^*(\omega)+\Delta u \in U$. 
\end{assumption}
In line with the Lemma \ref{L:est_bounds}, for the selected design parameters $c_{\delta}$, $c_{\gamma^*}$ and $c_{\gamma_g}\in (0,\gamma_g)$, we denote the minimal possible length of the training data sequence that still guarantees satisfaction of Lemma \ref{L:est_bounds} by $\tilde{N}$. We now have all the ingredients to state the theorem on the conditions under which the on-line application of Algorithm \ref{A:learning_algorithm} that uses on-line update of of some of the related tuning parameters results in a closed loop system that is finite gain stable. 
\begin{theorem}\label{T:rec_stabil}
Let the Assumptions \ref{A:bound}--\ref{A:contained} hold. If the parameters $\sigma$ and $\gamma_{\Delta, t}$ are selected such that \eqref{second_rule_sig} and \eqref{eq:other_bound} hold, if $\hat{S}^+_0\neq \emptyset$ and $x_0\in B_{\overline{x}}$. Then for any reference signal $r_t\in B_{\overline{r}},\forall t\geq 0$, it holds that $\hat{S}_t^+\neq \emptyset, \forall t\geq 0$ and the closed loop system obtained when Algorithm \ref{A:learning_algorithm} with the weight update equation as in \eqref{Eq:projection_algorithm2} is used is finite gain stable.
\end{theorem}
\begin{proof}
We will prove the  theorem by induction. But first, we note that the closed loop system obtained by using the approximate controller $f_t$ can be represented as:
\beq\label{Eq:treba1}
\begin{aligned}
x_{t\!+\!1}&\!=\!g(x_t,f_t(x_t,r_{t\!+\!1}))\!+\!e_t\!=\!g(x_t,f^*(x_t,r_{t\!+\!1}))\!+\!e_t\!+\!v_{t},\\
v_{t}&=g(x_t,f_t(x_t,r_{t\!+\!1}))-g(x_t,f^*(x_t,r_{t\!+\!1})).
\end{aligned}
\eeq
From Assumption  \ref{A:stabilizability} and the definition of the optimal inverse controller $f^*$ in \eqref{Eq:opt_ctrl}, it holds that: \beq\label{Eq:inner_loop}
||\bold{x}||_{\infty}\leq
\lambda_1^*||\bold{r}||_{\infty}+\lambda_2^*||\bold{v}||_{\infty}+\lambda_2^*||\bold{e}||_{\infty}+\beta^*, \eeq
where $\bold{v}=(v_1, v_2, \hdots)$. Moreover, we note
that from Assumptions \ref{A:lip} and \ref{A:contained}, it
follows that: \beq\label{Eq:F}
\begin{aligned}
||v_{t}||\!\leq\!\gamma_g |f_t(x_t,\!r_{t+1})\!-\!f^*(x_t,\!r_{t+1})|, \forall
(x_t,r_{t+1})\!\in\! B_{\overline{xr}}.
\end{aligned}
\eeq We now employ the inductive  argument to show that if $\hat{S}^+_0\neq\emptyset$ and $x_0\in B_{\overline{x}}$,
then $\hat{S}^+_t\neq \emptyset$ and $x_t\in
B_{\overline{x}},\forall t\geq 0$. The condition is satisfied for
$t=0$ by the Theorem assumption. Let us assume, for the sake of inductive argument, that $\hat{S}^+_k\neq \emptyset$ and $x_k\in B_{\overline{x}}, \forall
k\in [0,t-1]$. From this assumption and the way the weighting vector $a_t$ is updated in \eqref{Eq:projection_algorithm2}, it
follows that $a_k\in \hat{S}^+_k, \, \forall k\in[0,t-1]$. From Assumptions \ref{A:measurements} and \ref{A:contained}, it follows that $\omega_k \in X\times X$ and $u_k\in U, \forall k\in [-N,t-1]$. From the definition of $\hat{S}_t^+$ in \eqref{Eq:stability_stripv}, Assumptions \ref{A:bound}--\ref{A:lip} and Theorem \ref{T:milanese}, it then follows that: \beq\label{Eq:cond_sat}
|f_k(x_{k},r_{k+1})-f^*(x_{k},r_{k+1})|\leq
\gamma_{\Delta,k}||x_{k}||+\sigma, \forall k \in [0,t-1]. \eeq
From \eqref{Eq:F} and \eqref{Eq:cond_sat} it then holds that:
\beq\label{Eq:v} ||v_{k}||\leq
\gamma_g\gamma_{\Delta,k}||x_{k}||+\gamma_g\sigma, \forall k
\in [0,t-1]. \eeq From Lemma \ref{L:est_bounds} and the condition that $\overline{N}\geq N\geq \tilde{N}$, it follows that $\hat{\gamma}_{g,t}\geq \gamma_{g}-c_{\gamma_g}$, and hence from \eqref{Eq:time_varpar} it holds that $\gamma_{\Delta,t}<\frac{1}{\gamma_g\lambda_2^*}, \forall t\geq 0$. By using this and the fact that $\gamma_{\Delta, t}\leq \overline{\gamma}_{\Delta}$, we can show that if $\hat{S}_k^+ \neq \emptyset$ and $x_k\in B_{\overline{x}}, \forall k\in[0,t-1]$, then it holds that:
\beq\label{Eq:closed_loop}
\begin{aligned}
||\bold{x_t}||_{\infty}\leq&\frac{\lambda_1^*}{1-\gamma_g\lambda_2^*\overline{\gamma}_{\Delta}}||\bold{r_t}||_{\infty}+\frac{\lambda_2^*}{1-\gamma_g\lambda_2^*\overline{\gamma}_{\Delta}}||\bold{e_t}||_{\infty}\\&+\frac{\gamma_g\lambda_2^*\sigma+\beta^*}{1-\gamma_g\lambda_2^*\overline{\gamma}_{\Delta}}.
\end{aligned}
\eeq
From the definition of $\overline{f}_t$ and $\underline{f}_t$ in \eqref{Eq:bounds2} and $\overline{f}_c$ and $\underline{f}_c$ in \eqref{Eq:bounds_tv1}, it follows that $\overline{f}_t(\omega)\leq \overline{f}_c(\omega)$ and $\underline{f}_t(\omega)\geq \underline{f}_c(\omega), \forall \omega \in B_{\hat{x}\overline{r}},\forall t\geq 0$. Hence, it follows that:
\beq \label{Eq:satisfies}
\overline{f}_t(\omega)-\underline{f}_t(\omega)\leq \overline{f}_c(\omega)-\underline{f}_c(\omega),\forall \omega \in B_{\overline{xr}}.
\eeq  

From the definition of $\overline{x}$ in \eqref{eq:other_bound} it then follows that $\omega_t^+\in B_{\overline{x}}$, and hence from \eqref{Eq:satisfies} and \eqref{Eq:cond_sigma} it holds that:
\beq
-\gamma_{\Delta,t}-\sigma+\overline{f}_t(\omega_t^+)\leq \gamma_{\Delta,t}+\sigma+\underline{f}_t(\omega_t^+),
\eeq
and therefore $\hat{S}_t^+\neq \emptyset$. Repeating this inductive argumentation for all $t\geq 0$, it follows that $\hat{S}^+_t\neq \emptyset,\forall t\geq0$. In addition, \eqref{Eq:closed_loop} will hold for all $t\geq0$
which implies that the closed loop system is finite gain stable (see
e.g. Definition \ref{D:FGS}).\hfill$\blacksquare$
\end{proof}
Note that the Theorem \ref{T:rec_stabil} does not give any relation between the lower limit on the number of the collected initial training samples $\tilde{N}$ and the tuning parameters $c_{\delta}$, $c_{\gamma^*}$ and $c_{\gamma_g}$. These values need to be chosen by the control designer and should reflect his feeling of the quality with which the noise bound and the Lipschitz constants are estimated on the basis of the training data.


\begin{thebibliography}{1}

\bibitem{lorenzo}
  C. Novara, L. Fagiano and M. Milanese.
\newblock Direct feedback control design for nonlinear systems.
\newblock {\em  Automatica }, 49: 849--860, 213.

\bibitem{traub22}
 F. J. Traub and H. Wozniakowski.
\newblock A general theory of optimal algorithms
\newblock {\em  Academic Press, New York }, 1980.

\bibitem{nonlinear_milanese}
 M. Milanese and C. Novara.
\newblock Set Membership identification of nonlinear systems.
\newblock {\em Automatica}, 40:957--975, 2004.

\bibitem{Proj_3}
K. Slavakis and I. Yamada.
\newblock The adaptive projected subgradient method constrained by families of quasi-nonexpansive mappings and its application to online learning.
\newblock {\em  SIAM Journal of Optimization }, 23: 126--152, 2013.

\bibitem{koherence}
 C. Richard, J. C. M. Bermudesz and P. Honeine.
\newblock Online prediction of time series data with kernels.
\newblock {\em IEEE Transactions on Signal Processing}, 57: 1058--1067, 2009.

\bibitem{theodoridis_rev}
S. Theodoridis,  K. Slavakis and I. Yamada.
\newblock Adaptive learning in a world of projections.
\newblock {\em IEEE Signal Processing Magazine}, 28: 97--123, 2011.

\bibitem{lorenzo_archive}
L. Fagiano and C. Novara.
\newblock Identification of nonlinear controllers from data: theory and computation.
\newblock {\em arXiv:1309.1574}, 2013.


\end{thebibliography}
\end{document}